\documentclass[a4paper,UKenglish,cleveref, autoref, thm-restate]{lipics-v2021}

\graphicspath{{./graphics/}}
\usepackage{comment}
\usepackage{xcolor}
\usepackage[disable]{todonotes} %Toggle disable to eliminate comments
\usepackage{xspace}

\title{Determinization of One-Counter Nets} 

\author{ }{ }{}{}{}
\author{Shaull Almagor}{Department of Computer Science, Israel Institute of Technology, 3200002, Israel}{shaull@cs.technion.ac.il}{0000-0001-9021-1175}{}

\author{Asaf Yeshurun}{Department of Computer Science, Israel Institute of Technology, 3200002, Israel}{asafyeshurun@cs.technion.ac.il}{0000-0001-9149-1172}{}

\authorrunning{ }
\authorrunning{Shaull Almagor and Asaf Yeshurun} 

\Copyright{Asaf Yeshurun and Shaull Almagor} 



\category{} %optional, e.g. invited paper

\relatedversion{} %optional, e.g. full version hosted on arXiv, HAL, or other respository/website
\nolinenumbers %uncomment to disable line numbering

\EventEditors{John Q. Open and Joan R. Access}
\EventNoEds{2}
\EventLongTitle{42nd Conference on Very Important Topics (CVIT 2016)}
\EventShortTitle{CVIT 2016}
\EventAcronym{CVIT}
\EventYear{2016}
\EventDate{December 24--27, 2016}
\EventLocation{Little Whinging, United Kingdom}
\EventLogo{}
\SeriesVolume{42}
\ArticleNo{23}
\newtheorem*{example*}{Example}

\begin{document}

\newcommand{\tup}[1]{\langle #1 \rangle}

\newcommand{\cA}{\mathcal{A}}
\newcommand{\cB}{\mathcal{B}}
\newcommand{\cC}{\mathcal{C}}
\newcommand{\cG}{\mathcal{G}}
\newcommand{\cV}{\mathcal{V}}

\newcommand{\cD}{\mathcal{D}}
\newcommand{\cN}{\mathcal{N}}
\newcommand{\cM}{\mathcal{M}}

\newcommand{\bbN}{\mathbb{N}}
\newcommand{\bbZ}{\mathbb{Z}}

\newcommand{\lang}{\mathcal{L}}
\newcommand{\effect}[1]{\textsf{eff}(#1)}
\newcommand{\nadir}[1]{\textsf{nadir}(#1)}
\newcommand{\step}[2][]{\xrightarrow[#1]{#2}}

\newcommand{\loc}{\mathsf{Loc}}
\newcommand{\counters}{\mathsf{Z}}
\newcommand{\ops}{\mathsf{OP}}
\newcommand{\op}{\mathsf{op}}
\newcommand{\inc}{\mathsf{++}}
\newcommand{\dec}{\mathsf{--}}
\newcommand{\jz}{\mathsf{=0?}}
\renewcommand{\vec}[1]{\boldsymbol{#1}}

\newcommand{\sem}[1]{[\![#1]\!]}
\renewcommand{\phi}{\varphi}
\newcommand{\reach}{\psi_{\mathrm{Reach}}}
\newcommand{\lin}{\mathrm{Lin}}

\newcommand{\finreach}{\textsc{Finite-Reach}\xspace}
\newcommand{\finreachzero}{\textsc{0-Finite-Reach}\xspace}

\newcommand{\zerodet}{\texttt{0-Det}\xspace}
\newcommand{\existsdet}{\texttt{$\exists$-Det}\xspace}
\newcommand{\foralldet}{\texttt{$\forall$-Det}\xspace}
\newcommand{\unidet}{\texttt{Uniform-Det}\xspace}
\newcommand{\MCR}{\textrm{MCR}}

\newcommand{\shtodo}[1]{\todo[inline,color=cyan!25,size=\small]{SH: #1}}
\newcommand{\aytodo}[1]{\todo[inline,color=orange!25,size=\small]{AY: #1}}

\maketitle

\begin{abstract}
One-Counter Nets (OCNs) are finite-state automata equipped with a counter that is not allowed to become negative, but does not have zero tests. 
Their simplicity and close connection to various other models (e.g., VASS, Counter Machines and Pushdown Automata) make them an attractive model for studying the border of decidability for the classical decision problems.

The deterministic fragment of OCNs (DOCNs) typically admits more tractable decision problems, and while these problems and the expressive power of DOCNs have been studied, the determinization problem, namely deciding whether an OCN admits an equivalent DOCN, has not received attention.

We introduce four notions of OCN determinizability, which arise naturally due to intricacies in the model, and specifically, the interpretation of the initial counter value. 
We show that in general, determinizability is undecidable under most notions, but over a singleton alphabet (i.e., 1 dimensional VASS) one definition becomes decidable, and the rest become trivial, in that there is always an equivalent DOCN.

\end{abstract}
%\shtodo{General comment -- replace OCN with OCNs when plural}

\newpage

\section{Introduction}
\label{sec:intro}
One-Counter Nets (OCNs) are finite-state machines equipped with an integer counter that cannot decrease below zero and which cannot be explicitly tested for zero.

OCNs are closely related to several computational models: they are a test-free syntactic restriction of One-Counter Automata -- Minsky Machines with only one counter. Over a singleton alphabet, they are the same as 1-dimensional Vector Addition Systems with States, and if counter updates are restricted to $\pm 1$, they are equivalent to Pushdown Automata with a single-letter stack alphabet.

An OCN $\cA$ over alphabet $\Sigma$ accepts a word $w\in \Sigma^*$ from initial counter value $c\in \bbN$ if there is a run of $\cA$ on $w$ from an initial state to an accepting state in which the counter, starting from value $c$, does not become negative. Thus, for every counter value $c\in \bbN$ the OCN $\cA$ defines a language $\lang(\cA,c)\subseteq \Sigma^*$. We define the complement of a language $\lang(\cA,c)$ to be $=\overline{\lang(\cA,c)}=\Sigma^*\setminus\lang(\cA,c)$ 

OCNs are an attractive model for studying the border of decidability of classical decision problems. Indeed -- several problems for them lie delicately close to the decidability border. For example, OCN universality is decidable~\cite{hofman2018trace}, whereas paratemerized-universality (in which the initial counter is existentially quantified) is undecidable~\cite{almagor2020parametrized}. 

As is the case with many computational models, certain decision problems for  deterministic OCNs (DOCNs) are computationally easier than for nondeterministic OCNs (e.g., inclusion is undecidable for OCNs, but is in NL for DOCNs~\cite{hofman2018trace}.  Universality is Ackermannian for OCNs, but is in NC for DOCNs~\cite{almagor2020parametrized}).
While decision problems for DOCNs have received some attention, the \emph{determinization} problem for OCNs, namely deciding whether an OCN admits an equivalent DOCN, has (to our knowledge) not been studied. 
Apart from the theoretical interest of OCN determinization, which would yield a better understanding of the model, it is also of practical interest: OCNs can be used to model properties of concurrent systems, so when an OCN can be determinized, automatic reasoning about correctness becomes easier. 

\paragraph*{OCN Determinization}
Recall that the language $\lang(\cA,c)$ of an OCN $\cA$ depends on the initial counter $c$, so $\cA$ essentially defines a family of languages. Thus, it is not clear what we mean by ``equivalent DOCN''. 
We argue that the definition of determinization depends on the role of the initial counter $c$. To this end, we identify four notions of determinization for an OCN $\cA$,  as follows. 
\begin{itemize}
    \item In \zerodet, we ask whether there is a DOCN $\cD$ such that $\lang(\cA,0)=\lang(\cD,0)$.
    \item In \existsdet, we ask whether there exist $c\in \bbN$ and a DOCN $\cD$ such that $\lang(\cA,c)=\lang(\cD,0)$.
    \item In \foralldet, we as whether for every $c\in \bbN$, there is a DOCN $\cD$ such that $\lang(\cA,c)=\lang(\cD,0)$.
    \item In \unidet, we ask whether there is a DOCN $\cD$ such that for every $c\in \bbN$ we have $\lang(\cA,c)=\lang(\cD,c)$.
\end{itemize}
The motivation for each of the problems depends, intuitively, on the interpretation of the initial counter, and on the stage at which the equivalent DOCN is computed: consider a factory, comprising a number $c$ of machines, that should be controlled by a DOCN.
The design of the controller, however, uses a nondeterministic OCN, whose initial counter value would be the number $c$ of machines. 

Now, if $c$ is fixed, then a corresponding DOCN would only have to match this specific counter, so \zerodet is needed\footnote{Any fixed initial counter $c$ can be simulated with a $0$ counter, by initially adding $c$.}. If the factory is being planned, and any amount of machines can be purchased in advance, then \existsdet is suitable. If the OCN controller needs to be reused in many factories, each with a different number of machines, then \foralldet is required, since we are assured a correct DOCN would be available for each factory. Finally, if the number of machines varies each day (e.g., if the machines are virtual, or have downtime), then we need the DOCN to correctly handle any number of them, so \unidet is suitable.

\paragraph*{Paper Organization and Contribution}
In this paper, we study the decidability of the four determinization notions. In \cref{sec:determinization} we examine the relation between the notions, and demonstrate that no pair of them coincide. In \cref{sec:general} we show that \zerodet,\existsdet, and \foralldet are generally undecidable. For \unidet, we are not able to resolve decidability, but we do show an Ackermannian lower bound. 

In order to recover some decidability, we turn to the fragment of OCNs over a singleton alphabet (1-dimensional VASS). There, we show that \zerodet,\existsdet, and \foralldet become trivial (i.e., they always hold), whereas \unidet becomes  decidable.  We conclude with a discussion and future work in \cref{sec:discussion}.

Technically, our undecidability results use reductions from two different models -- one from the model of Lossy Counter Machines~\cite{mayr2003undecidable,schnoebelen2010lossy}, and one from a careful analysis of recent results about OCNs~\cite{almagor2020parametrized}. For the decidability results, the triviality of \zerodet,\existsdet, and \foralldet is shown using standard tools for OCNs, whereas the decidability of \unidet requires some machinery from the theory of low-dimensional VASS and Presburger Arithmetic, as well as some basic linear algebra and number theory.

\paragraph*{Related Work}
The determinization problem we consider in this work assumes that the deterministic target model is also that of OCNs. An alternative approach to simplifying a nondeterministic OCN is to find an equivalent deterministic finite automaton, if one exists. This amounts to deciding whether the language of an OCN is regular. This problem was shown to be undecidable for OCNs in~\cite{valk1981petri}.  Interestingly, the related problem of regular separability was shown to be in PSPACE in~\cite{czerwinski2017regular}.
A related result in~\cite{finkel2006omega} describes a determinization procedure for ``unambiguous blind counter automata'' over infinite words, to a Muller counter machine.

From a different viewpoint, determinization is a central problem in quantitative models, which can be thought of as counter automata where the counter value is the output, rather than a Boolean language acceptor. The decidability of determinization for Tropical Weighted Autom49ata is famously open~\cite{buchsbaum2000determinization,kirsten2009deciding} with only partial decidable fragments~\cite{kirsten2009deciding,klimann2004deciding}. A slightly less related model is that of discounted-sum automata, whose determinization has intricate connections to number theory~\cite{boker2014exact}.

%Due to lack of space, some proofs appear in the Appendix.

% \begin{table}[h!]
% \centering
% \begin{tabular}{||c || c c||} 
%  \hline
%   & \textbf{Unary Alphabet} & \textbf{General Alphabet} \\ [0.5ex] 
%  \hline\hline
%  \textbf{$\zerodet$} & trivial (always applies) & Undecidable \\ 
%  \textbf{$\foralldet$} & trivial (always applies) & Undecidable \\
%  \textbf{$\existsdet$} & trivial (always applies) & Undecidable \\
%  \textbf{$\unidet$} & Some reasonable complexity & $\Omega$(Ackermann) \\ [1ex]
%  \hline 
% \end{tabular}
% \newline
% \caption{The complexity of the OCN-determinizability problem for all variations discussed}
% \label{table:main}
% \end{table}

\section{Preliminaries}
\label{sec:prelims}
A \emph{one-counter net} (OCN) is a finite automaton whose transitions are labelled both by letters and by integer weights. Formally, an OCN is 
a tuple $\cA=\tup{\Sigma,Q,s_0,\delta,F}$ where $\Sigma$ is a finite alphabet,  $Q$ is a finite set of states, $s_0\in Q$ is the initial  state, $\delta\subseteq Q\times\Sigma\times \bbZ\times  Q$ is the transition relation, and $F\subseteq Q$ are the accepting states. We say that an OCN is \emph{deterministic} if for every $s\in Q,\sigma\in \Sigma$, there is at most one transition $(s,a,e,s')$ for any $e\in \bbZ$ and $s'\in Q$.

For a transition $t=(s,a,e,s')\in\delta$ we define $\effect{t}= e$
to be its (counter) \emph{effect}.
%, and write $\norm{\delta}$ for the largest absolute effect among all transitions.
%By the \emph{underlying automaton} of an OCN
%we mean the NFA obtained from the OCN by disregarding the transition effects.

A \emph{path} in the OCN is a sequence $\pi=(s_1,\sigma_1,e_1,s_2)(s_2,\sigma_2,e_2,s_3)\dots (s_{k},\sigma_k,e_k,s_{k+1})\in\delta^*$. Such a path $\pi$ is a \emph{cycle} if $s_1=s_{k+1}$, and is a \emph{simple cycle} if no other cycle is a proper infix of it.
We say that the path $\pi$ \emph{reads} the word $\sigma_1\sigma_2\dots \sigma_k\in\Sigma^*$. The \emph{effect} of $\pi$ is $\effect{\pi}=\sum_{i=1}^k e_i$, and its \emph{nadir}, denoted $\nadir{\pi}$, is the minimal effect of any prefix of $\pi$ (note that the nadir is non-positive, since $\effect{\epsilon}=0$).
%\shtodo{Do we ever use the definitions of effect and nadir anymore?}
%and it is\emph{accepting} if 
% $s_{k+1}\in F$.
% Its $\effect{\pi}\eqdef \sum_{i=1}^k e_i$ is the sum of its transition effects .
% Its \emph{height} is the maximal effect of any prefix and, similarly, its
% \emph{depth} is the inverse of the minimal effect of any prefix.

% \medskip
% An OCN naturally induces an infinite-state labelled transition system in which each
A \emph{configuration} of an OCN is a pair $(s,v)\in Q\times \bbN$ comprising a state and a non-negative integer. 
For a letter $\sigma\in \Sigma$ and configurations $(s,v),(s',v')$ we write $(s,v)\step{\sigma}(s',v')$ if there exists $d\in \bbZ$ such that $v'=v+d$ and $(s,\sigma,d,s')\in \delta$.

A \emph{run} of $\cA$ from initial counter $c$ on a word $w=\sigma_1 \cdots \sigma_k\in\Sigma^*$ is a sequence of configurations 
$\rho=(q_0,v_0),(s_1,v_1),\ldots ,(s_k,v_k)$
such that $v_0=c$ and for every $1\le i\le k$ it holds that $(s_{i-1},v_{i-1})\step{\sigma_i}(s_i,v_i)$. Since configurations may only have a non-negative counter, this enforces that the counter does not become negative.

Note that every run naturally induces a path in the OCN. For the converse, a path $\pi$ induces a run from initial counter $c$ iff $c\ge -\nadir{\pi}$ (indeed, the minimal initial counter needed for traversing a path $\pi$ is exactly $-\nadir{\pi}$). 
We extend the definitions of effect and nadir to runs, by associating them with the corresponding path.

The run $\rho$ is \emph{accepting} if $s_k\in F$, and we say that $\cA$ \emph{accepts} $w$ with initial counter $c$ if there exists an accepting run of $\cA$ on $w$ from initial counter $c$. We define $\lang(\cA,c)=\{w\in \Sigma^*: \cA\text{ accepts $w$ with initial counter }c\}$.
Observe that OCNs are monotonic -- if $\cA$ accepts $w$ from counter $c$, it also accepts it from every $c'\ge c$. Thus, $\lang(\cA,c)\subseteq \lang(\cA,c')$ for $c'\ge c$.

\section{OCN Determinization}
\label{sec:determinization}

In this section we examine the relationship between the four determinization notions. For brevity, we use the same term for the decision problems and the properties they represent, e.g., we say ``$\cA$ is \zerodet'' if $\cA$ has an equivalent DOCN under \zerodet.

Our first observation is that the definitions are comparable in their strictness:
\begin{observation}
\label{obs:definitionImplications}
Consider an OCN $\cA$. If $\cA$ is \unidet, then $\cA$ is \foralldet, if $\cA$ is \foralldet, then $\cA$ is \zerodet, and if $\cA$ is \zerodet, then $\cA$ is \existsdet.  
\end{observation}
Next, we prove that none of the definitions coincide. Following~\cref{obs:definitionImplications}, it suffices to prove the following.
\begin{lemma}
\label{lem:definitionsNotCoincide}
There exist OCNs $\cA,\cB,\cC$ such that $\cA$ is \existsdet but not \zerodet, $\cB$ is \zerodet but not \foralldet, and $\cC$ is \foralldet but not \unidet.
\end{lemma}
\begin{proof}[Proof (sketch)]
The OCNs $\cA,\cB,\cC$ are depicted in \cref{fig:separation}. We demonstrate the intuition on $\cC$, see \cref{app:definitionsNotCoincide} for the complete proof. To show that $\cC$ is \foralldet, we observe that for initial counter $0$, we can omit the $(\#,-1)$ transition, thus obtaining an equivalent DOCN. For initial counter $c\ge 1$ we have that $\lang(\cC,c)=\#\cdot \{a,b\}^*$ is regular and thus has a DOCN.

We claim $\cC$ is not \unidet. An equivalent DOCN $\cD$ with $k$ states, starting from initial counter $0$, must accept the word $\#a^{k+1}b^{k+1}$. It is easy to show that upon reading $b^{k+1}$ it must make a negative cycle. This, however, causes some word of the form $\#a^{k+1}b^m$ not to be accepted even with counter 1, which means $\lang(\cD,1)\neq \lang(\cC,1)$. 
\end{proof}

\begin{figure}[ht]
\captionsetup[subfigure]{justification=centering}
\begin{subfigure}[b]{.5\linewidth}
\centering \includegraphics[page=1]{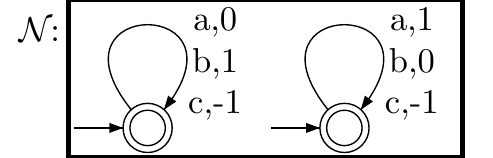}
\caption{Gadget OCN $\cN$.}\label{fig:gadget}
\end{subfigure}%
\begin{subfigure}[b]{.5\linewidth}
\centering \includegraphics[page=2]{graphics/detSeparation-pics.pdf}
\caption{\existsdet but not \zerodet.}\label{fig:existsDetNotZeroDet}
\end{subfigure}\\
\begin{subfigure}[b]{.4\linewidth}
\centering \includegraphics[page=3]{graphics/detSeparation-pics.pdf}
\caption{\zerodet but not \foralldet.}\label{fig:zeroDetnotForallDet}
\end{subfigure}%
\begin{subfigure}[b]{.5\linewidth}
\centering \includegraphics[page=4]{graphics/detSeparation-pics.pdf}
\caption{\foralldet but not \unidet.}\label{fig:forallDetnotUniDet}
\end{subfigure}
\caption{Examples separating the determinization notions.}\label{fig:separation}
\end{figure}

\section{Lower Bounds for Determinization}
\label{sec:general}
In this section we prove lower bounds for the four determinization decision problems.
%- $\zerodet, \foralldet, \existsdet$ and \unidet, for OCNs over a general, unrestricted alphabet. 
We show that \zerodet, \foralldet, and \existsdet are undecidable, while for \unidet we show an Ackermannian lower bound, and its decidability remains an open problem.

We start by introducing \emph{Lossy Counter Machines} (LCMs)~\cite{mayr2003undecidable,schnoebelen2010lossy}, from which we will obtain some undecidability results. Intuitively, an LCM is a Minsky counter machine, whose semantics are such that counters may arbitrarily decrease at each step.
Formally, an LCM is $\cM=\tup{\loc,\counters,\Delta}$ where $\loc=\{\ell_1,\ldots,\ell_m\}$ is a finite set of \emph{locations}, $\counters=(z_1,\ldots, z_n)$ are $n$ counters, and $\Delta\subseteq \loc\times \ops(\counters)\times \loc$, where $\ops(\counters)=\counters\times \{\inc,\dec,\jz\}$.

A \emph{configuration} of $\cM$ is $\tup{\ell,\vec{a}}$ where $\ell\in \loc$ and $\vec{a}=(a_1,\ldots,a_n)\in \bbN^{\counters}$. There is a \emph{transition} $\tup{\ell,\vec{a}}\to \tup{\ell,\vec{b}}$ if there exists $\op\in \ops$ and either:
\begin{itemize}
    \item $\op=c_k\inc$ and $b_k\le a_k+1$ and $b_j\le a_j$ for all $j\neq k$, or
    \item $\op=c_k\dec$ and $b_k\le a_k-1$ and $b_j\le a_j$ for all $j\neq k$, or
    \item $\op=c_k\jz$ and $b_k=a_k=0$ and $b_j\le a_j$ for all $j\neq k$.
\end{itemize}
Since we only require $\le$ on the counter updates, the counters nondeterministically decrease at each step.

A \emph{run} of $\cM$ is a finite sequence $\tup{\ell_1,\vec{a_1}}\to \tup{\ell_2,\vec{a_2}}\to\ldots\to \tup{\ell_r,\vec{a_r}}$. Given a configuration $\tup{\ell,\vec{a}}$, the \emph{reachability set} of $\tup{\ell,\vec{a}}$ is the set of all configurations reachable from $\tup{\ell,\vec{a}}$ via runs of $\cM$.

In~\cite{schnoebelen2010lossy}, it is shown that the problem of deciding whether the reachability set of a configuration is finite, is undecidable. A slight modification of this problem (see~\cref{app:finFromZeroUndecidable}) yields the following.

\begin{lemma} \label{lem:finFromZeroUndecidable}
The following problem, dubbed \finreachzero, is undecidable:
Given an LCM $\cM$ and a location $\ell_0$, decide whether the reachability set of $\tup{\ell_0,(0,\ldots, 0)}$ is finite.
\end{lemma}

\subsection{Undecidability of \zerodet} 
\label{subsec:zerodet}
% We show that $\zerodet$ is undecidable using a reduction from $\finreachzer$: Given an LCM $M=(Loc,C,\Delta)$ and a configuration $\sigma=(q,c_1 \cdots c_k)$, decide whether the reachability set of $\sigma$ finite. This problem is known to be undecidable \cite{bla}. 

% It is technically easier to work with a variant of $\finreach$ in which the LCM starts from $(q_0,0 \cdots 0)$. We will refer to this variant as $\finreachzero$.

% \begin{lemma} \label{lem:finFromZeroUndecidable}
% $\finreachzero$ is undecidable.
% \end{lemma}
We show that \zerodet is undecidable using a reduction from \finreachzero. Intuitively, given an LCM $\cM$ and a location $\ell_0$, we construct an OCN $\cA$ that accepts, from initial counter 0, all the words that do not represent runs of $\cM$ from $\tup{\ell_0,(0,\ldots,0)}$.

In order for the OCN $\cA$ to verify the illegality of a run, it guesses a violation in it. Control violations, i.e., illegal transitions between locations, are easily checked. In order to capture counter violations, $\cA$ must find a counter whose value in the current configuration is \emph{smaller} than in the next iteration (up to $\pm 1$ for $\inc$ and $\dec$ commands). 
This, however, cannot be done by an OCN, since intuitively an OCN can only check that the later number is smaller. To overcome this, we encode runs in reverse, as follows.

Consider an LCM $\cM=\tup{\loc,\counters,\Delta}$ with $\loc=\{\ell_1,\ldots,\ell_m\}$ and $\counters=(z_1,\ldots, z_n)$. We encode a configuration $\tup{\ell,(a_1,\ldots,a_n)}$ over the alphabet $\Sigma=\loc\cup \counters$ as $\ell\cdot z_1^{a_1}\cdots z_n^{a_n}\in \Sigma^*$. We then encode a run by concatenating the encoding of its configurations.

For a word $w=\sigma_1\cdots \sigma_k\in \Sigma^*$, let $w^R=\sigma_k\cdots \sigma_1$ be its \emph{reverse}, and for a language $L\subseteq \Sigma^*$, define its reverse to be $L^R=\{w^R:w\in L\}$.

We now define 
$
L_{\cM,\ell_0}=\{w\in \Sigma^*: w \text{ encodes a run of $\cM$ from }\tup{\ell_0,(0,\ldots,0)}\}.
$

% \begin{definition}[LCM Alphabet]
% Let $M=(\{l_1 \cdots l_m\},\{c_1 \cdots c_k\},\Delta)$ be an LCM. We define $M$'s alphabet to be $\Sigma=\left\{\left\{q\right\}_{q \in Q} \cup \left\{a_i\right\}_{1 \leq i \leq k}\right\}$. 
% \end{definition}

% \begin{definition}[Configuration Encoding]
% Let $M=(\{l_1 \cdots l_m\},\{c_1 \cdots c_k\},\Delta)$ be an LCM, let $\Sigma$ be its alphabet and let $\rho=(q,c_1 \cdots c_k)$ be a configuration of $M$ (that may or may not be reached by a legal run). We define $\rho$'s encoding to be $w=qa_1^{c_1} \cdots a_k^{c_k} \in \Sigma^*$.
% \end{definition}

% \begin{definition}[Run Encoding]
% Let $M=(\{l_1 \cdots l_m\},\{c_1 \cdots c_k\},\Delta)$ be an LCM, let $\Sigma$ be its alphabet and let $\pi=\rho_0 \cdots \rho_n$ be a run of $M$, such that the configuration encoding of $\{\rho_i\}_{i=0}^n$ is $\{w_i\}_{i=0}^n$ respectively. We define $\pi$'s encoding as $w=w_0 \cdots w_n$. 
% \end{definition}

% \begin{definition}[Mirror Language]
% Let $\Sigma=\{\sigma_1 \cdots \sigma_l\}$ be an alphabet, and let $L \subseteq \Sigma^*$. We define the mirror of $L$, denoted $L^R$, to be $\left\{w=\sigma_{i1} \cdots \sigma_{im}, \sigma_{im} \cdots \sigma_{i1} \in L\right\}$. 
% \end{definition}

%We are now ready to prove undecidability of $\zerodet$. The proof's core is captured by the following Lemma:

We are now ready to describe the construction of $\cA$.
\begin{lemma} \label{lem:lcmToOcnConstruction}
Given an LCM $\cM$ and a location $\ell_0$, we can construct an OCN $\cA$ such that $\lang(\cA,0)=\overline{L_{\cM,\ell_0}^R}$. 
% Let $M=(\{l_1 \cdots l_m\},\{c_1 \cdots c_k\},\Delta)$ be an LCM with a single initial location $l_1$, let $\Sigma$ be its alphabet, and let $L=\left\{w, w\text{ encodes a legal run of } M, \text{ starting from } \rho_0=(l_1, 0 cdots 0)\right\}$. We can construct, in polynomial time, an OCN $\cA=(\Sigma,Q,s_0,\delta,F)$ such that $\lang(\cA,0)=[L^R]^C$. 
\end{lemma}
\begin{proof}[Proof sketch:]
We construct $\cA$ such that it accepts a word $w$ iff $w^R$ does not describe a run of $\cM$ from $\tup{\ell_0,(0,\ldots, 0)}$.

As mentioned above, $\cA$ reads $w$ and guesses when a violation would occur, where control violations are relatively simple to spot, by directly encoding the structure of $\cM$ in $\cA$.

In order to spot counter violations, namely two consecutive configurations $\tup{\ell,(a_1,\ldots, a_n)}$ and $\tup{\ell',(a'_1,\ldots, a'_n)}$ such that some $a'_i$ is too large compared to its counterpart $a_i$ (how much larger is "too large" depends on $\cM$'s transitions), $\cA$ reads a configuration $\ell\cdot z_1^{a_1}\cdots z_n^{a_n}$ and increments its counter to count up to $a_i$, if it guesses that $z_i$ is the counter that violates the transition. 
Assume for simplicity that the command in the transition does not involve counter $z_i$, then upon reading the next configuration $\ell'\cdot z_1^{b_1}\cdots z_n^{b_n}$, $\cA$ decrements its counter while reading $z_i$, so that the counter value is $a_i-b_i$. Then, $\cA$ takes another transition with counter value $-1$. Since the configuration is reversed, if this is indeed a violation, then $a_i>b_i$ (since the counters are lossy), so $a_i-b_i-1\ge 0$, and $\cA$ accepts. Otherwise, $a_i\le b_i$, so this run of $\cA$ cannot proceed.

In \cref{app:lcmToOcnConstruction} we give the complete details of the construction.
\end{proof}

The correctness of the construction is proved in the following lemmas.

\begin{lemma} \label{lem:finReachThendet}
Consider an LCM $\cM$ and a location $\ell_0$, and let $\cA$ be the OCN constructed in \cref{lem:lcmToOcnConstruction}. If $(\cM,\ell_0)$ is in $\finreachzero$, then $\cA$ is $\zerodet$.
\end{lemma}

\begin{proof}

Assume the reachability set of $\tup{\ell_0,(0 \ldots 0)}$ is finite under $\cM$. We will prove that $\overline{\lang(\cA,0)}$ is regular. This results in $\lang(\cA,0)$ being regular, since regular languages are closed under complementation, hence there is a DFA that accepts $\lang(\cA,0)$, which is also a DOCN that accepts $\lang(\cA,0)$ (regardless of initial counter value).

Since the reachability set of $\tup{\ell_0,(0 \ldots 0)}$ is finite, all counter values are bounded in all legal runs. Let $m$ be an upper bound for all counters in all legal runs. We now construct a DFA $\cD$ that accepts all words that correspond to legal runs of $\cM$.

Intuitively, our construction method is the following: The states of $\cD$ are of the form $\tup{\ell,a_{1} \ldots a_{k}, b_{1} \ldots b_{k}}$ where $\ell$ is a state of $\cM$ and $0 \leq a_{1} \ldots a_{k}, b_{1} \ldots b_{k} \leq m$. 

Within the state, $a_{1} \ldots a_{k}$ represent the current counter values (that are already fully known), while $b_{1} \ldots b_{k}$ represent the counters of the next configuration, which $\cD$ is in the process of reading. We also add an initial state $\tup{\bot, \bot \ldots \bot, 0 \ldots 0}$, and all states of the form $\tup{\bot, \bot \ldots \bot, b_{1} \ldots b_{k}}$. 

From state $\tup{\ell,a_{1} \ldots a_{k}, 0 \ldots 0}$, $\cD$ can read a word of the form $z_k^*\cdot z_{k-1}^*\cdots z_1^*$, while updating the $b_1,\ldots,b_k$ components to accumulate the counters (up to $m$). When a letter $\ell'$ is read from state $\tup{\ell,a_{1} \ldots a_{k} , b_{1} \ldots b_{k}}$, then $\cD$ has a transition to $\tup{\ell',b_{1} \ldots b_{k} , 0 \ldots 0}$ iff the configuration $\tup{\ell,a_{1} \ldots a_{k}}$ can be reached from $\tup{\ell',b_{1} \ldots b_{k}}$ in a single transition of $\cM$. 

%For each state of the form $\tup{q,a_{1} \ldots a_{k}, 0 \ldots 0}$, including the initial state, we add a transition by reading $z_k$ to $\tup{q,a_{1} \ldots a_{k} , 0 \ldots 1}$, then by reading another $z_k$ to $\tup{q,a_{1} \ldots a_{k} , 0 \ldots 2}$, etc. up to the upper bound $m$. From a state of the form $\tup{q,a_{1} \ldots a_{k} , 0, \ldots, x}$ we add a transition by reading $a_{k-1}$ to $\tup{q,a_{1}, \ldots, a_{k} , 0 \ldots 1,x}$, etc. From states of the form $\tup{q,a_{1} \ldots a_{k} , b_{1} \ldots b_{k}}$ we add a transition by reading $q'$ to $\tup{q',b_{1} \ldots b_{k} , 0 \ldots 0}$ iff the configuration $\tup{q,a_{1} \ldots a_{k}}$ can be reached from the configuration $\tup{q',b_{1} \ldots b_{k}}$ through a single transition in $\cM$. 
Finally, the only accepting state of $\cD$ is $\tup{\ell_0,0 \ldots 0 , 0 \ldots 0}$. 
We give the formal construction of $\cD$ in \cref{app:finReachThendet}.

The correctness of the construction is immediate, as $\cD$ simply tracks legal runs, and the only technicality is dealing with the encoding. Thus, $\lang(\cA,0)$ is regular, and $\cA$ is $\zerodet$.

%It is straightforward to show that the construction satisfies the requirements, therefore $\lang(\cA,0)$ is regular, and $\cA$ is $\zerodet$.
\end{proof}

\begin{lemma} \label{lem:detThenFinReach}
Consider an LCM $\cM$ and a location $\ell_0$, and let $\cA$ be the OCN constructed in \cref{lem:lcmToOcnConstruction}. If $(\cM,\ell_0)$ is not in $\finreachzero$, then $\cA$ is not $\zerodet$.
\end{lemma}

\begin{proof}
%We will show that if $M$ does not satisfy $\finreachzero$, $\cA$ is not $\zerodet$.

Assume the reachability set of $\tup{\ell_0,(0 \ldots 0)}$ under $\cM$ is infinite, and assume by way of contradiction that $\cA$ has a deterministic equivalent $\cD$ with $d$ states.
Observe that for every word $u\in \Sigma^*$, the run of $\cD$ does not end due to the counter becoming negative. Indeed, we can always concatenate some $\lambda\in \Sigma^*$ such that $u\lambda$ does not correspond to a run, and is hence accepted by $\cD$, so the run on $u$ must be able to continue reading $\lambda$. We call this property of $\cD$ \emph{positivity}.

Since the reachability set of $\tup{\ell_0,(0 \ldots 0)}$ is infinite, there exists a counter of $\cM$, w.l.o.g $z_1$, that can take unbounded values (in different runs). Let $w$ be a word corresponding to a run of $\cM$ that ends with the value of $z_1$ being $N$ for some $N>d$. 
We can then write $w=a_k^* \cdots a_1^N \ell a_k^* \cdots a_1^{N'} \ell'\rho$ such that $\rho$ represents the reverse of a legal prefix of a run of $\cM$, and $N'$ satisfies $N' \geq N-1$, since no single transition of $\cM$ can increase a counter by more than one (but $N'$ can be arbitrarily large).

Since $w$ corresponds to a legal run of $\cM$, $\cA$ (and therefore $\cD$) does not accept $w$. By the positiviy of $\cD$, its run on $w$ ends in a non-accepting state. 

Since $N > d$, $\cD$ goes through a cycle $\beta$ when reading $a_1^{N}$. We pump the cycle $\beta$ to obtain a run of $\cD$ on a word $w''=a_k^* \cdots a_1^{N+t} q a_k^* \cdots a_1^{N'}q'\rho$ for some $t \in \bbN$ that satisfies $N+t > N'+1$. 
Again, by the positivity of $\cD$, the run cannot end due to the counter becoming negative, so it ends in the same non accepting state as the run on $w$.
However, $w''$ does not represent a legal run of $M$, since $N+t > N'+1$, therefore $w'' \in \lang(\cA,0)$, which contradicts $\lang(\cA,0)=\lang(\cD,0)$. 
\end{proof}

Combining \cref{lem:lcmToOcnConstruction,lem:finReachThendet,lem:detThenFinReach}, we conclude the following.
\begin{theorem}
\label{thm:zerodetUndecidable}
$\zerodet$ is undecidable for OCNs over a general alphabet.
\end{theorem}

\subsection{Undecidability of $\foralldet$}
The undecidability of \foralldet follows from that of \zerodet.
\begin{theorem}
\label{thm:foralldetUndecidable}
\foralldet is undecidable.
\end{theorem}
\begin{proof}
We show a reduction from \zerodet. Given an OCN $\cA=\tup{\Sigma,Q,s_0,\delta,F}$, we construct an OCN $\cB=\tup{\Sigma',Q',q_0,\delta',F'}$ as illustrated in figure \ref{fig:forall_det_construction}. Formally, the states of $\cB$ are $Q'=Q\cup \{q_0,q_{\text{All}}\}$, the initial state is $q_0$, its alphabet is $\Sigma'=\Sigma \cup \{\#\}$ such that $\# \notin \Sigma$, its accepting states are $F'=F\cup \{q_{\text{All}}\}$, and its transition relation is $\Delta'=\Delta \cup \{(q_0,\#,-1,q_{\text{All}}),(q_0,\#,0,s_0)\} \cup \{(q_{All},\sigma,0,q_{\text{All}}) : \sigma \in \Sigma'\}$.

\begin{figure}[ht] 
\centering
\includegraphics[page=1]{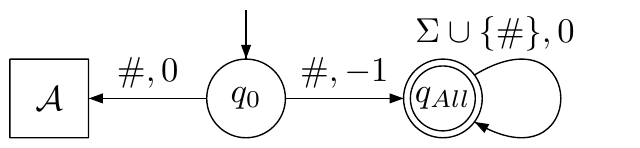}
\caption{The OCN $\cB$.}
\label{fig:forall_det_construction}
\end{figure}

We claim that $\cA$ is $\zerodet$ iff $\cB$ is $\foralldet$. For the first direction, assume $\cA$ is $\zerodet$. Thus, $\lang(\cB,0)=\#\cdot \lang(\cA,0)$ has an equivalent DOCN. Since $\lang(\cB,k)=\#\Sigma'^*$ (which has a DOCN) for all $k \geq 1$, $\cB$ is $\foralldet$.

Conversely, assume $\cA$ is not $\zerodet$. Since the transition $(q_0,\#,-1,q_{All})$ cannot be taken by $\cB$ with initial counter value 0, $\lang(\cB,0)=\left\{\#w\right\}_{w \in \lang(\cA,0)}$, hence $\cB$ is not $\zerodet$ (since a DOCN for $\lang(\cB,0)$ would easily imply a DOCN for $\lang(\cA,0)$). Thus, $\cB$ is not $\foralldet$.  
\end{proof}

\subsection{Undecidability of $\existsdet$}
We show the undecidability of $\existsdet$ with a reduction from the halting problem for two-counter (Minsky) machines (2CM). 
Technically, we rely on a construction from~\cite{almagor2020parametrized}, which reduces the latter problem to the ``parameterized universality'' problem for OCN. For our purpose, the reader need not be familiar with Minsky Machines, as it suffices to know that their halting problem is undecidable~\cite{minsky1967computation}.
We start the reduction with the following property.
\begin{theorem}[\cite{almagor2020parametrized}]
\label{thm:parameterizedUniversality}
Given a 2CM $\cM$, we can construct an OCN $\cB$ over alphabet $\Sigma\cup \{\#\}$ with $\#\notin \Sigma$ such that the following holds:
\begin{itemize}
    \item If $\cM$ halts, there exists $c\in \bbN$ such that $\lang(\cB,c)=\Sigma^*$,
    \item If $\cM$ does not halt, then for every $c\in \bbN$ there exists a word $w_c\in (\Sigma\cup \{\#\})^*$ such that every run of $\cB$ on $w_c$ ends in a state from which reading any word of the form $\#^*$ does not lead to an accepting state.
\end{itemize}
\end{theorem}
We can now proceed with the reduction to \existsdet.
\begin{theorem}
\label{thm:existsdetUndecidable}
\existsdet is undecidable.
\end{theorem}
\begin{proof}
We reduce the halting problem for 2CM to \existsdet. Given a 2CM $\cM$, we start by constructing the OCN $\cB$ as per \cref{thm:parameterizedUniversality}. We augment $\cB$ to work over the alphabet $\Sigma'=\Sigma\cup\{\#,\$,\%\}$, where $\$,\%\notin \Sigma$, by fixing the behaviour of $\$$ and $\%$ to be identical to $\#$.

Next, consider the gadget OCN $\cG$ depicted in \cref{fig:exists_determinization_sub_component}. A similar argument to the proof of \cref{lem:definitionsNotCoincide} (specifically, \cref{fig:gadget}), shows that $\cG$ does not have an equivalent DOCN for any initial counter $c$.

% Our reduction is a simple variant of the one applied by  \cite{almagor2020parametrized} to show undecidability of initial value universality for OCNs over a general alphabet. 

% Let $\cA$ be the OCN illustrated in figure \ref{fig:exists_determinization_sub_component}. Formally, $\cA=(\Sigma,Q,S_0,\delta,F)$ such that $\Sigma=\{\$,\#\,\%\}$, $Q=\{q_1,q_2\}$, $S_0=\{q_1,q_2\}$, $\delta=\{(q_1,\$,-1,q_1),(q_1,\%,0,q_1),(q_1,\#,1,q_1)\} \cup \\ \{(q_2,\$,0,q_2),(q_2,\%,-1,q_2),(q_2,\#,1,q_2)\}$.

\begin{figure}[ht] 
\caption{Gadget OCN $\cG$ for Theorem~\ref{thm:existsdetUndecidable}.}
\centering
\includegraphics[page=2]{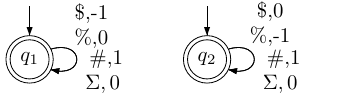}
\label{fig:exists_determinization_sub_component}
\end{figure}

We now obtain a new OCN $\cA$ by taking the union of $\cB$ and $\cG$ (i.e. placing them ``side by side'').
% $\cA'=(\Sigma',Q',S_0',\Delta',F')$, such that:
% \begin{itemize}
%     \item $\Sigma'=\Sigma\cup\{\$,\%\}$.
%     \item $Q'=Q\cup\{q_1,q_2\}$.
%     \item $S_0'=S_0\cup\{q_1,q_2\}$.
%     \item $F'=F\cup\{q_1,q_2\}$.
%     \item $\delta\subseteq\delta'$.
%     \item For all $(q,\#,z,q')\in \delta$, $\{(q,\$,z,q'),(q,\%,z,q')\}\subseteq \delta'$.
%     \item $\{(q_1,\$,-1,q_1),(q_1,\%,0,q_1),(q_1,\#,1,q_1),(q_2,\$,0,q_2),(q_2,\%,-1,q_2), \\ (q_2,\#,1,q_2) \subseteq \delta'$.
%     \item For all $\sigma \in \Sigma$, $\{(q_1,\sigma,0,q_1),(q_2,\sigma,0,q_2)\}\subseteq \delta'$.
% \end{itemize}
We claim that $\cM$ halts iff $\cA$ is $\existsdet$. 

If $\cM$ halts, by \cref{thm:parameterizedUniversality} there exists an initial counter $c$ such that $\lang(\cB,c)=\left\{\Sigma \cup \{\#\}\right\}^*$. Since in $\cB$ the letters $\$$ and $\%$ behave like $\#$,  we have that $\lang(\cA,c)=\Sigma'^*$, so $\cA$ is $\existsdet$. 

If $M$ does not halt, then again by \cref{thm:parameterizedUniversality}, for every $c\in\bbN$ There exists a word $w_c$ such that $w_c\notin \lang(\cB,c)$, and such every run of $\cB$ on $w_c$ end in a state from which reading $\#^*$ (and hence any word from $\{\#,\%\,\$\}^*$) does not lead to an accepting state.
Now assume by way of contradiction that $\cA$ has a deterministic equivalent $\cD$ with $k$ states for initial counter $c$. 
$\cA$ accepts $w_c$ with the runs of $\cG$, since $w_c$ does not contain $\$$ or $\%$. Thus, $\cD$ accepts $w_c$ with initial counter $0$.
In addition, $\cA$, and therefore $\cD$, both accept $w_c'=w_c\#^{k+1-j}\%^{k+1}\$^{k+1}$ where $j$ is the number of occurrences of \#'s in $w_c$. Using the fact that $\cG$ does not have an equivalent DOCN, we can now reach a contradiction with similar arguments as the proof of \cref{lem:definitionsNotCoincide} (\cref{fig:gadget}). 
\end{proof}

\subsection{A Lower Bound for \unidet}
Unfortunately, as of yet we are unable to resolve the decidability of \unidet. In this section, we show that \unidet is Ackermann-hard, and in particular NON-ELEMENTARY.

\begin{theorem}
\label{thm:unidetAckermannn}
\unidet is Ackermann-hard.
\end{theorem}
\begin{proof}
%Here we show that OCN uniform-determinizability is $\Omega$(Ackermann). We prove this result using a reduction from OCN universality with initial counter value 0, which has been shown to be $\Omega$(Ackermann) 
We show a reduction from the OCN universality problem with initial counter 0, shown to be Ackermann-hard in~\cite{hofman2018trace}.

Consider an OCN $\cA=\tup{Q,\Sigma,s_0,\delta,F}$. We construct an OCN $\cB=\tup{Q',\Sigma',q_0,\delta',F'}$ as depicted in \cref{fig:uniformGeneralHelper} (for $\#,\$ \notin \Sigma$). 
%%%% Long version:
% Formally,
% \begin{itemize}
%     \item $Q'=Q \cup \{q_0,q_1,q_{All}\}$.
%     \item $\Sigma'=\Sigma \cup \{\#,\$\}$.
%     \item $F'=F \cup \{q_{All}\}$.
%     \item $\delta'=\delta \cup \left\{(q_0,\#,0,s_0), (q_0,\#,0,q_1), (q_0,\#,-1,q_{All}), (q_1,\$,0,q_{All}),\right\} \cup \left\{ (q,\$,0,q_{All}) \right\}_{q \in Q}$.
% \end{itemize}

\begin{figure}[ht]
\centering
\includegraphics[page=3]{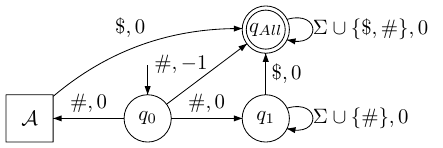}
\caption{The OCN $\cB$ in the proof of \cref{thm:unidetAckermannn}.}
\label{fig:uniformGeneralHelper}
\end{figure}
We claim that $\lang(\cA,0)=\Sigma^*$ iff $\cB$ is $\unidet$.

Assume $\lang(\cA,0)=\Sigma^*$, then $\lang(\cB,c)=\{\#w:w \in \Sigma'^*\}$ for every counter value $c$. Indeed, every word starting with \# can be accepted by $\cB$ with initial counter value 0 either through $\cA$, if it does not contain $\$$, or in $q_{\text{All}}$ if it does. However, every word not starting with \# cannot be accepted by $\cB$ for any initial counter value. Thus, $\cB$ is \unidet.

Conversely, if $\lang(\cA,0) \neq \Sigma^*$, let $w \notin \lang(\cA,0)$. Assume by way of contradiction that there exists a deterministic OCN $\cD$ that is uniform-equivalent to $\cB$.

$\#w \notin \lang(\cB,0)$, so $\#w \notin \lang(\cD,0)$. Moreover, the run of $\cD$ on $\#w$ cannot end in a non-accepting state, since $\#w\in \lang(\cB,1)=\lang(\cD,1)$. Thus, the run of $\cD$ on $\#w$ terminates due to the counter becoming negative. However, this is a contradiction, since  $\#w\$\in \lang(\cB,0)=\lang(\cD,0)$.
We conclude that $\cB$ is not \unidet.
\end{proof}

\section{Singleton Alphabet}
\label{sec:singleton}
We now turn to study OCNs over a singleton alphabet denoted $\Sigma=\{\sigma\}$ throughout.
%In this section we discuss determinization of OCNs over a singleton alphabet. 
%We show that all OCNs under this restriction are $\zerodet$, $\existsdet$ and $\foralldet$, while not always $\unidet$. Nevertheless, we prove that $\unidet$ is decidable, and provide a construction method for a uniform deterministic equivalent when relevant.

We start by briefly introducing Presburger Arithmetic (PA)~\cite{Haase2018,presburger1929uber}. We refer the reader to~\cite{Haase2018} for a detailed survey. 
PA is the first-order theory of integers with addition and order $\mathrm{FO}(\bbZ,0,1,+,<)$, and it is a decidable logic. 

There is an important connection between PA and semilinear sets: 
for a \emph{basis vector} $\vec{b}\in \bbZ^d$ and a set of \emph{periods} $P=\{\vec{p_1}\ldots \vec{p_k}\}\subseteq \bbZ^d$, we define the \emph{linear set} $\lin(\vec{b},P)=\{\vec{b}+\lambda_1\vec{p_1}+\ldots +\lambda_k \vec{p_k}:\lambda_i\in \bbN \text{ for all }1\le i\le k\}$.
Then, a \emph{semilinear} set is a finite union of linear sets.

A fundamental theorem about PA~\cite{ginsburg1964bounded} shows that that for every PA formula $\phi(\vec{x})$ with free variables $\vec{x}$, the set $\sem{\phi}=\{\vec{a}:\vec{a}\models \phi(\vec{x})\}$ is semilinear, and the converse also holds -- every semilinear set is PA-definable.  

\shtodo{Check the structure of our PA formulas, to see what we can say about complexity.}

%We start by defining the Nadir Mapping of an OCN - an abstract mathematical object that fully captures its language for all initial counter values.

Consider an OCN $\cA$ over $\Sigma=\{\sigma\}$. For every word $\sigma^n$, either $\sigma^n$ is not accepted by $\cA$ for any counter value, or there exists a minimal counter value $c$ such that $\sigma^n\in \lang(\cA,c')$ iff  $c'\ge c$. We can therefore fully characterize the language of $\cA$ on any counter value using the \emph{Minimal Counter Relation\footnote{We remark that $\MCR(\cA)$ is in fact the graph of a partial function. For convenience of working with PA, we stick with the relation notation.}} (\MCR), defined as
\[\MCR(\cA)=\left\{(n,c)\subseteq \bbN^2, c \text{ is the minimal integer such that } \sigma^n \in \lang(\cA,c) \right\}.\]
% \begin{definition}[Nadir Mapping]
% Let $\cA$ be an OCN over $\Sigma=\{a\}$. We define its nadir mapping to be $M=\left\{(n,c)\subseteq \bbN^2, c \text{ is the minimal integer that satisfies } a^n \in L_c(A) \right\}$. 
% \end{definition}

%Note that, since $\Sigma^*= \left\{a^i | i \in \bbN\right\}$, $\cA$'s nadir mapping contains complete information as to what is required of an initial configuration for $\cA$ to accept a word $w \in \Sigma^*$, if at all it can be accepted.  

We start by showing that $\MCR(\cA)$ is semilinear.
\begin{lemma}
\label{lem:MCRisSemilinear}
Consider an OCN $\cA$ over $\Sigma=\{\sigma\}$, then $\MCR(\cA)$ is effectively semilinear.
\end{lemma}
\begin{proof}
We prove the claim using well-known and deep results about low-dimensional VASS.
A 2D-VASS is (for our purposes\footnote{Usually, OCNs are defined as 1D-VASS, not the other way around.}) identical to an OCN over $\Sigma=\{\sigma\}$, but has two counters (both need to be kept non-negative). 
Formally, a 2D VASS is $\cV=\tup{Q,s_0,\delta,F}$, where $\delta\subseteq Q\times\bbZ^2\times Q$. The semantics are similar to OCNs, acting separately on the two counters, as follows. 
%We refer to e.g.~\cite{leroux2004flatness} for complete definitions.\shtodo{Add in the appendix?}
A \emph{configuration} of $\cV$ is $(q,(c_1,c_2))$ where $q\in Q$ and $(c_1,c_2)\in \bbN^2$ are the counter values, and a \emph{run} is a sequence of configurations $(q_1,(c^1_1,c^1_2)),\ldots,(q_k,(c^k_1,c^k_2))$ that follow according to $\delta$, i.e., for every $1\le i <k$ we have that $(q_i,(c^{i+1}_1-c^i_1,c^{i+1}_2-c^i_2),q_{i+1})\in \delta$. We denote $(q_1,(c^1_1,c^1_2))\step{\cV}(q_k,(c^k_1,c^k_2))$ if such a run exists.

In~\cite{leroux2004flatness}, it is proved that given a 2D-VASS, we can effectively compute a PA formula $\reach(q,x_1,x_2,q',y_1,y_2)$ such that $\sem{\reach(q,x_1,x_2,q',y_1,y_2)}=\{(q,c_1,c_2,q',d_1,d_2): (q,c_1,c_2)\step{\cV}(q',d_1,d_2)\}$ (the states $q,q'$ are encoded as variables taking values in $\{1,\ldots, |Q|\}$).

Observe that $\reach$ does not encode information about the length of the run, whereas $\MCR$ does require it. On the other hand, $\reach$ works for 2D-VASS, whereas we only need an OCN (i.e., 1D-VASS). We therefore proceed as follows.
Given the OCN $\cA=\tup{\{\sigma\},Q,s_0,\delta,F}$, we construct a 2D-VASS $\cV=\tup{Q,s_0,\delta',F}$ where $\delta'=\{(q,(v_1,1),q'): (q,\sigma,v_1,q')\in \delta\}$. That is, $\cV$ works exactly as $\cA$, with the second counter marking the length of the run. In particular, fixing $\cA$ to start from $q_0$ with initial counter $c$, the formula $\reach(q_0,c,0,q',d,n)$ is satisfied by $(q',d,n)$ such that $\cA$ can reach state $q'$ upon reading $\sigma^{n}$ from initial counter $c$, ending with counter $d$.

Thus, we can now define $\MCR$ as follows. Consider the formula
$\phi(x,y)= \bigvee_{q'\in F} \exists v\in \bbN, \reach(q_0,0,y,q',x,v)$, then $\sem{\phi(x,y)}=\MCR$. Indeed, by the above, $\phi(x,y)$ is satisfied by $(n,c)$ iff there exists an accepting state $q'\in F$ and a counter value $v$ such that $\cA$ can reach $q'$ upon reading $\sigma^n$ starting from counter $c$. 
We can then require $c$ to be minimal by defining
$\theta(x,y)=\phi(x,y)\wedge \forall y',\ y'<y\to \neg \phi(x,y)$. We then have that $\sem{\theta(x,y)}=\MCR$.
Thus, $\MCR$ is definable in PA, and by~\cite{ginsburg1964bounded} it is effectively semilinear.
\end{proof}

% \begin{remark}
% \label{rmk:MCRSemilinearAlternativeProof}
% \cref{lem:MCRisSemilinear} can also be proved directly using the fact that 1D-VASS are \emph{flat} (see~\cite{leroux2004flatness,blondin2015reachability}). This complicates the proof, but allows a generalization of it to 2D-VASS instead of OCNs, as well as allows complexity analysis. 
% We favour the proof above for elegance.
% \shtodo{rephrase this if we have complexity results}
% \end{remark}

\subsection{Decidability of $\unidet$ over Singleton Alphabet}
In this subsection we prove that $\unidet$ is decidable for OCN over a singleton alphabet, and we can construct an equivalent DOCN, if one exists. 
Our characterization of \unidet is based on its $\MCR$, and specifically on two notions for subsets of $\bbN^2$ (applied to $\MCR$). Consider a set $S\subseteq \bbN^2$. We say that $S$ is \emph{increasing} if it is the graph of an increasing partial function. That is, for every $(n_1,c_1),(n_2,c_2)\in S$, if $n_1\le n_2$ then $c_1\le c_2$, and if $n_1=n_2$ then $c_1=c_2$.
Next, we say that $S$ is \emph{$(N,k,d)$-Ultimately Periodic} for $N,k,d\in \bbN$ if for every $n \geq N, (n,x) \in S $ iff $(n+k,x+d) \in S$. We say that $S$ is (effectively) ultimately periodic if it is $(N,k,d)$-ultimately periodic for some (effectively computable) parameters $N,k,d\in \bbN$.

The main technical result of this section is the following.
%%%%Short version:
% \begin{theorem}
% Consider an OCN $\cA$ over $\Sigma=\{\sigma\}$, then $\cA$ is \unidet iff $\MCR(\cA)$ is increasing, and if $\MCR(\cA)$ is increasing then it is also effectively ultimately periodic.
% \end{theorem}

%%%%Long version:
\begin{theorem}
\label{thm:unidetEquivalentConditions}
Consider an OCN $\cA$ over $\Sigma=\{\sigma\}$, then the following are equivalent:
\begin{enumerate}
    \item $\MCR(\cA)$ is increasing.
    \item $\MCR(\cA)$ is increasing and effectively ultimately periodic.
    \item $\cA$ is \unidet, and we can effectively compute an equivalent DOCN.
\end{enumerate}
\end{theorem}
We prove \cref{thm:unidetEquivalentConditions} in the remainder of this section. We start with a technical lemma concerning the implication $1\implies 2$.
\begin{lemma}
\label{lem:semilinearAndIncreadingThenPeriodic}
Consider an effectively semilinear set $S\subseteq \bbN^2$. If $S$ is increasing, then $S$ is effectively periodic.
\end{lemma}
\begin{proof}
Since $S$ is effectively semilinear, then by~\cite{ginsburg1964bounded} we can write $S=\bigcup_{i=1}^M \lin(\vec{b_i},P_i)$ where $\vec{b_i}\in \bbN^2$ and $P_i\subseteq \bbN^2$ for every $1\le i\le M$.
Moreover, by~\cite{ginsburg1964bounded,Haase2018}, we can assume that each $P_i$ is a linearly-independent set of vectors.
\vspace{-1.25em}
\subparagraph*{All periods are singletons: }
We start by claiming that since $S$ is increasing, then $|P_i|\le 1$ for every $1\le i\le M$. 
Indeed, assume $(n_1,c_1),(n_2,c_2)\in P_i$, and denote $\vec{b_i}=(a,b)$, then by the definition of a linear set, for every $\lambda_1,\lambda_2\in \bbN$ we have that $(a,b)+\lambda_1 (n_1,c_1)+\lambda_2 (n_2,c_2)\in S$. Setting $\lambda_1=0$ and $\lambda_2=n_1$, we have that $(a+n_1 n_2, b+ n_1 c_2)\in S$, and setting $\lambda_1=n_2$ and $\lambda_2=0$, we have that $(a+ n_2 n_1, b+ n_2 c_1)\in S$. Observe that $a+n_1n_2=a+n_2n_1$, and since $S$ is increasing, this implies $b+n_1c_2=b+n_2c_1$, that is $n_1c_2=n_2c_1$. 
It follows that $n_2(n_1,c_1)=n_1(n_2,c_2)$, but $P_i$ is linearly independent, so it must hold that $(n_1,c_1)=(n_2,c_2)$, so $|P_i|\le 1$.

Thus, we can in fact write $S=\bigcup_{i=1}^M \lin(\vec{b_i},\{\vec{p_i}\})$ where $\vec{b_i},\vec{p_i}\in \bbN^2$ (note that if $P_i=\emptyset$ we now take $\vec{p_i}=(0,0)$). 
For every $1\le i\le M$, denote $\vec{b_i}=(a_i,b_i)$ and $\vec{p_i}=(p_i,r_i)$. 
\vspace{-1.25em}
\subparagraph*{All Periods have the same first component: }
We now claim that we can restrict all periods to have the same first component. That is, we can compute $\gamma\in \bbN$ and 
write $S=\bigcup_{j=1}^K\lin((\alpha_j,\beta_j),\{(\gamma,\eta_j)\})$.

Indeed, take $\gamma=\mathrm{lcm}(\{p_i\}_{i=1}^M)$, we now ``spread'' each linear component $\lin((a_i,b_i),\{p_i,r_i\})$ by changing the period to $(\gamma, \frac{\gamma}{p_i} r_i)$, and compensating by adding additional linear sets with the same period and offset basis, to capture the ``skipped'' elements. In \cref{app:semilinearAndIncreadingThenPeriodic} we describe the construction in general. We illustrate it here with an example: 
\begin{example*}
Let $S=\lin((1,0),(4,8)) \cup \lin((2,1),(6,12))$. We have $\gamma=\textrm{lcm}(4,6)=12$. 
We write $\lin((1,0),(4,8))=\lin((1,0),(12,24)) \cup \lin((5,8),(12,24)) \cup \lin((9,16),(12,24))$, the intuition being that instead of e.g., a $(4,8)$ period, we have a $(12,24)$ period, and we add different basis vectors to fill the gaps, so the new basis vectors are $(5,8)$ and $(9,16)$, where the next basis vector $(13,24)$ is already captured by $(1,0)+(12,24)$.

Similarly, we write $\lin((2,1),(6,12))=\lin((2,1),(12,24)) \cup \lin((8,13),(12,24))$. 
Overall we get $S=\bigcup_{\vec{b} \in B}\lin(\vec{b},(12,24))$ for $B=\left\{(1,0),(5,8),(9,16),(2,1),(8,13)\right\}$.

Note that $S$ in the example is increasing, and that we actually end up with the same period vector, not just on the first component. As we show next, this is not a coincidence.
\end{example*}

%\shtodo{either add here, or say that it's in the appendix.}
\vspace{-1.25em}
\subparagraph{All Periods are the same: }
Finally, we claim that we now have $\eta_i=\eta_{j}$ for every $1\le i,j\le K$, so that in fact all the periods are the same vector $(\gamma,\eta)$. Indeed, 
Assume by way of contradiction that $\eta_j < \eta_i$ for some $1 \le i,j \leq K$. Now, let $y \in \bbN$ be large enough so that $\alpha_i \leq \alpha_{j} + y \cdot \gamma$, and let $x \in \bbN$ be large enough so that (given $y$): $\beta_{i}+x \cdot \eta_i > \beta_{j} + y \cdot \eta_j + x \cdot \eta_j$. 

We now have that $(\alpha_i,\beta_i) + x \cdot (\gamma,\eta_i) \in S$ and $(\alpha_j,\beta_j) + (x+y) \cdot (\gamma,\eta_j) \in S$, which contradicts $S$ being increasing, since $\alpha_{i} \leq \alpha_{j} + y \cdot \gamma$ and therefore $\alpha_{i} + x \cdot \gamma \leq \alpha_{j} + (y + x) \cdot \gamma$, but also $\beta_{i}+x \cdot \eta_i > \beta_{j} + (y + x) \cdot \beta_j$. 

Thus, we can now write $S=\bigcup_{j=1}^K\lin((\alpha_j,\beta_j),\{(\gamma,\eta)\})$
\vspace{-1.25em}
\subparagraph{$S$ is effectively ultimately periodic: }
Let $\alpha_{\max}=\max\{\alpha_j\}_{j=1}^{K}$, we claim that $S$ is $(\alpha_{\max},\gamma,\eta)$-ultimately periodic.
Let $n\ge \alpha_{\max}$, then $(n,c)\in S$ for some $c\in \bbN$ iff $(n,c)=(\alpha_i+\gamma\cdot m,\beta_i+\eta\cdot m)$ for some $1\le i\le K$ and $m\in \bbN$. This happens iff $(n+\gamma,c+\eta)\in S$, since $(n+\gamma,c+\eta)=(\alpha_i+\gamma\cdot (m+1),\beta_i+\eta\cdot (m+1))$.

Finally, observe that all the constants in the proof are effectively computable.
\end{proof}

% The main claim we develop throughout this section is that an OCN $N$ over a singleton alphabet $\Sigma=\{a\}$ is $\unidet$ iff its nadir mapping is an ascending set. 
% Next we define another important characteristic of a set $S$ - $d$-periodic. Intuitively, this characteristic captures the notion that $S$ has a periodic nature for large enough $a_i$'s. For example, the set $S=\left\{(0,0),(1,0),(2,0),(3,1),(4,5),(5,6),(6,10),(7,11),(8,15) \cdots \right\}$ is 5-periodic, since from the fourth element onwards it carries a cyclic behaviour. A counter example would be the set $S'=\left\{(0,0),(1,1),(2,3),(3,6),(4,10),(5,15),(6,21),(7,28) \cdots \right\}$, which is not d-periodic for any $d \in \bbN$.

\shtodo{we might want to put back the example I commented above.}

We now turn to the implication $2\implies 3$ of \cref{thm:unidetEquivalentConditions}.

\begin{lemma} 
\label{lem:periodicThanUniDet}
Consider an OCN $\cA$ over $\Sigma=\{a\}$. If $\MCR(\cA)$ is increasing and ultimately periodic, then $\cA$ is $\unidet$, and we can effectively compute it.
\end{lemma}
\begin{proof}
Assume $\MCR(\cA)$ is $(N,k,d)$-ultimately periodic.
%
% Throughout this proof, it is best to view $M$ as a partial function $f:\bbN \rightarrow \bbN$, such that $f(n)=c$ iff $(n,c) \in M$. \par
% Before we describe the construction of a deterministic equivalent $\cD$ directly, we complete $f$ to obtain a function $f':\bbN \rightarrow \bbN$ that is defined on its entire domain. Intuitively, if $f(k)$ is not defined for some $k \in \bbN$, we define $f(k)$ to be $f(k')$ for $k'$ being the closest integer $k' < k$ for which $f(k')$ is defined. If no such $k'$ exist, we have $f(k)=0$.
%
% Formally, we define: 
% $f'(n)=
% \begin{cases}
% c, f(n)=c \\
% c', f(n')=c', \text{ and } n' \text{ is the largest integer } n' < n \text{ for which } f(n') \text{ is defined } \\
% 0, f(n') \text{ is undefined for all } n' \leq n.
% \end{cases}
% $
%
We start by completing $\MCR(\cA)$ to a (full) function $f:\bbN\to \bbN$ as follows: set $f(0)=0$, and for $n>0$ inductively define $f(n)=c$ if $(n,c)\in \MCR(\cA)$, or $f(n)=f(n-1)$ otherwise. That is, $f$ matches $\MCR(\cA)$ on its domain, and remains fixed between defined values. Observe that there is no violation in defining $f(0)=0$, since if $(0,c)\in \MCR(\cA)$, then $c=0$, as the empty word requires a minimal counter of $0$ to be accepted.
%
%Once we have defined $f'$ we are ready to discuss the construction formally. 

We now use $f$ to obtain a DOCN $\cD$ as depicted in \cref{fig:uniformEquivIllustration}. 
 Formally, we construct $\cD=\tup{\{\sigma\},Q,q_0,\delta,F}$ as follows. 
%The states are $Q=\left\{q_i\right\}_{i=1}^{N+k-1}$, and the transitions are $\delta=\left\{(q_i,a,f(i)-f(i+1),q_{i+1})\right\}_{i=1}^{N+k-2} \cup \left\{(q_{N+k-1},a,f(N)+d-f(N+k-1),q_N)\right\}.$ The accepting states are then $F=\{q_i:(i,f(i))\in \MCR(\cA),\ 1\le i\le N+k-1\}$.
\begin{itemize}
    % \item $\Sigma=\{\sigma\}$
    \item $Q=\left\{q_i\right\}_{i=1}^{N+k-1}$.
    \item $\delta=\left\{(q_i,a,f(i)-f(i+1),q_{i+1})\right\}_{i=1}^{N+k-2} \cup \left\{(q_{N+k-1},a,f(N)+d-f(N+k-1),q_N)\right\}$.
    \item $F=\{q_i:(i,f(i))\in \MCR(\cA),\ 1\le i\le N+k-1\}$.
\end{itemize}

Observe that since $f$ is increasing (as $\MCR(\cA)$ is increasing), the weight of all transitions in $\cD$ is non-positive.

\begin{figure}[h]
\centering
\includegraphics[page=4]{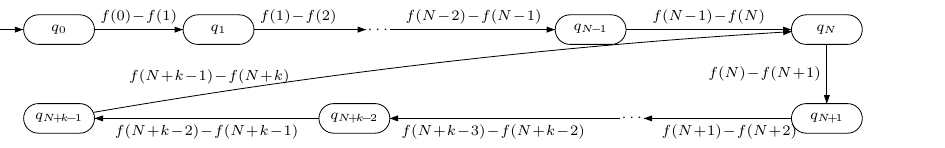}
\caption{An illustration of the construction method for a uniform-deterministic-equivalent of an OCN $\cA$, given $f$. Accepting states are not mentioned in the illustration.}
\label{fig:uniformEquivIllustration}
\end{figure}

We claim that for every $c$, $\lang(\cA,c)=\lang(\cD,c)$. To show this, observe that for every $n\in \bbN$ we have that the sum of weights along $n$ consecutive transitions of $\cD$ (ignoring the OCN semantics) is exactly $-f(n)$. In particular, if $\sigma^n\in \lang(\cA,c)$, then $(n,c')\in \MCR(\cA)$ for some $c'\le c$ and $f(n)=c'$. Indeed, this is trivial for $n\le N+k-1$, and for $n>N+k-1$ this follows immediately from $(N,k,d)$-ultimate periodicity.

Thus, if $\sigma^n\in \lang(\cA,c)$ then there exists $c'\le c$ such that $(n,c')\in  \MCR(\cA)$ it follows that with initial counter $c$, $\cD$ can traverse $n$ transitions. Moreover, the state reached is accepting, since $(n,c')\in \MCR(\cA)$, so $\sigma^n\in \lang(\cD,c)$.

Conversely, if $\sigma^n\in \lang(\cD,c)$ then $c\ge f(n)$ and $(n,f(n))\in \MCR(\cA)$, thus, $\sigma^n\in \lang(\cA,c)$.

Finally, observe that the construction is computable given the parameters of ultimate periodicity.
\end{proof}

We now address the implication $3\implies 1$.
\begin{lemma}
\label{lem:unidetImpliesIncreasing}
Consider an OCN $\cA$ over $\Sigma=\{\sigma\}$. If $\cA$ is \unidet, then $\MCR(\cA)$ is increasing.
\end{lemma}
\begin{proof}
Let $\cD$ be a DOCN such that $\lang(\cA,c)=\lang(\cD,c)$ for every $c$, and let $(n_1,c_1), (n_2,c_2) \in \MCR(\cA)$ with $n_1\le n_2$. Assume by way of contradiction that $c_1>c_2$, then $\sigma^{n_2}\in \lang(\cD,c_2)$, but $\sigma^{n_1}\notin \lang(\cD,c_2)$. It follows that the run of $\cD$ on $\sigma^{n_1}$ must end in a non-accepting state starting from counter value $c_2$ (i.e., the counter does not become negative). But then the same run is taken from counter value $c_1$, so $\sigma^{n_1}\notin \lang(\cD,c_1)$, which is a contradiction.
\end{proof}

By \cref{lem:MCRisSemilinear}, $\MCR(\cA)$ is semilinear. Thus, if $\MCR(\cA)$ is increasing, then by \cref{lem:semilinearAndIncreadingThenPeriodic} it is also effectively ultimately periodic. This completes the implication $1\implies 2$, and the implications $2\implies 3$ and $3\implies 1$ are immediate from \cref{lem:periodicThanUniDet,lem:unidetImpliesIncreasing}, respectively. This completes the proof of \cref{thm:unidetEquivalentConditions}.
% We now have all the needed implications.
% \begin{proof}[Proof of \cref{thm:unidetEquivalentConditions}]
% $1\implies 2$: assume $\MCR(\cA)$ is increasing. By~\cref{lem:MCRisSemilinear}, $\MCR(\cA)$ is effectively semilinear, so by~\cref{lem:semilinearAndIncreadingThenPeriodic} $\MCR$ is also effectively ultimately periodic.\\
% $2\implies 3$: follows immediately from \cref{lem:periodicThanUniDet}.\\
% $3\implies 1$: follows immediately from \cref{lem:unidetImpliesIncreasing}.
% \end{proof}

Finally, having characterized \unidet, we can show its decidability.
\begin{theorem}
\unidet is decidable.
\end{theorem}
\begin{proof}
Consider an OCN $\cA$. By~\cref{thm:unidetEquivalentConditions}, it suffices to show that it is decidable whether $\MCR(\cA)$ is increasing. By~\cref{lem:MCRisSemilinear}, we can compute a PA formula $\theta(n,c)$ such that $\sem{\theta}=\MCR(\cA)$. 
We now state the assertion that $\MCR(\cA)$ is not increasing as the following sentence in PA:
$\chi=\exists n_1,n_2,c_1,c_2,\ n_1<n_2\wedge c_1>c_2 \wedge \theta(n_1,c_1)\wedge \theta(n_2,c_2)$.

Since PA is decidable, we can decide whether this sentence holds, so we are done.
\end{proof}

% \begin{theorem}
% Let $\cA$ be an OCN over $\Sigma=\{a\}$ and let $M$ be its nadir mapping. Then $\cA$ is $\unidet$ iff $M$ is ascending.  
% \end{theorem}

% \begin{proof}
% Assume $M$ is ascending. From Lemma \ref{lem:nadirMapSemilinear} we have that $M$ is semilinear. Since $M$ is both semilinear and ascending, then from lemma \ref{lem:ascendingAndSemiPeriodic} $M$ is periodic, so by lemma \ref{lem:periodicThanUniDet} we have that $\cA$ is $\unidet$. 

% Conversely, if $M$ is not ascending, let $(n_1,c_1), (n_2,c_2) \in M$ such that $n_1 \leq n_2$ and $c_1 > c_2$. Assume by negation that there exists a deterministic OCN $\cD$ equivalent to $\cA$ for all initial counter values. $\cA$ accepts both $a^{n_1}$ and $a^{n_2}$ for large enough counter values, therefore the paths in $\cD$ of both $a^{n_1}$ and $a^{n_2}$ end in accepting states. The nadir of $\cD$ when reading $a^{n_1}$ is $-c_1$, therefore the run of $\cD$ on $a^{n_2}$ starting from initial counter value $c_2$ ends due to a counter violation during the first $n_1$ steps, which contradicts $(n_2,c_2) \in M$.\par
% This completes the proof.
% \end{proof}

\subsection{Triviality of $\zerodet$, $\foralldet$, $\existsdet$}
%Once we have proved that, given an OCN $\cA$, its nadir map $M$ is semilinear, 
We now turn to study the remaining notions of determinization for singleton alphabet. 
\begin{theorem}
\label{thm:singletonTrivialDet}
Consider an OCN $\cA$ over $\Sigma=\{\sigma\}$, then $\cA$ is \foralldet, \zerodet, and \existsdet.
\end{theorem}
\begin{proof}
By~\cref{obs:definitionImplications}, it is enough to prove that $\cA$ is \foralldet. To this end, recall that by~\cref{lem:MCRisSemilinear}, $\MCR(\cA)$ is PA definable by a formula $\phi(n,c)$.

For every initial counter value $c$, define $\phi_{\le c}(n)=\bigvee_{i=0}^{c}\phi(n,i)$, then $\sem{\phi_{\le c}(n)}=\{n: \cA \text{ accepts }\sigma^n\text{ with initial counter }c\}$. Then, we can write $\lang(\cA,c)=\{\sigma^{m}: m\in \sem{\phi_{\le c}(n)}\}$. 

It is folklore that a singleton-alphabet language whose set of lengths is semilinear, is regular. We bring a short proof of this for completeness:
Let $S=\bigcup_{i=1}^{k}\lin(c_i,p_i)\subseteq \bbN$ be a semilinear set (by assuming that the periods are linearly independent, it follows each has a single number), and let $L_S=\{\sigma^k: k\in S\}$. For every $i$, the language $\{a^k | k \in \lin(c_i,p_i)\}$ can be defined by the regular expression $r_i=\sigma^{c_i}(\sigma ^{p_i})^*$. So $L_S$ is defined by the regular expression $r=r_1+ \cdots +r_k$. 

Thus, for every $c\in \bbN$, we have that $\lang(\cA,c)$ is regular, and in particular is recognized by a DOCN, so $\cA$ is \foralldet, and we are done.
\end{proof}

\section{Discussion and Future Work}
\label{sec:discussion}
In this work, we introduce and study notions of determinization for OCNs. We demonstrate that the notions, while comparable in strictness, are distinct both from a conceptual perspective, having different motivations, as well as from a technical perspective: the mathematical tools needed to analyze them vary.

The most obvious direction for future work is resolving the decidability status of \unidet. Note that \unidet bears some similarities to the determinization problem for tropical automata, in that in both cases we need to match exactly the set of counters (or weights). As the latter problem is famously open, it could well be that \unidet is similarly difficult.

Finally, we do not give complexity bounds for the decidability result for \unidet in the singleton case. This follows from the black-box nature of our proof. Delving into the details of the construction in~\cite{leroux2004flatness}, or using methods such as those in~\cite{10.1145/3464794}, and combining them with our specific setting may yield some bounds.

\bibliography{bibliography}

\appendix

\section{Proofs}
\label{app:determinization}

\subsection{Proof of \cref{lem:definitionsNotCoincide}.}
\label{app:definitionsNotCoincide}
\subsubsection{$\cA$ is $\existsdet$, but not $\zerodet$}
\label{app:existsButNotZero}

We define formally $\cA=\tup{\{q_0,q',q'',q_5\},\{a,b,c,\#\},q_0,\delta_{\cA},\{q',q'',q_5\}}$, for:

$\delta_{\cA}=\tup{(q_0,\#.-5,q_5),(q_0,\#.0,q'),(q_0,\#.0,q''),(q',a,1,q'),(q',b,0,q')} \cup \\
\left\{(q',c,-1,q'),(q'',a,0,q''),(q'',b,1,q''),(q'',c,-1,q''),(q_5,a,0,q_5),(q_5,b,0,q_5),(q_5,c,0,q_5)\right\}$.

$\cA$ is $\existsdet$, since $\lang(\cA,k)=\Sigma^*$ for $k \geq 5$. Now, assume by way of contradiction that $\cA$ is $\zerodet$, and let $\cD$ be a deterministic OCN with $n\in \bbN$ states that satisifies $\lang(\cA,0)=\lang(\cD,0)$.
We now define $w=\#c^{n+1}a^{n+1}b^{n+1}$. throughout the run of $\cD$ on $w$, $\cD$ travels through a cycle $\beta_1$ when reading $a^{n+1}$, and a cycle $\beta_2$ when reading $b^{n+1}$. If the cumulative costs of both $\beta_1$ and $\beta_2$ are non-negative, then $\cD$ accepts $w'=\#c^{n+1}a^{N}b^{N}$ for arbitrarily large $N\in \bbN$, which contradicts $\lang(\cA,0)=\lang(\cD,0)$. Otherwise, the cumulative cost of either $\beta_1$ or $\beta_2$ is negative, w.l.o.g $\beta_1$. In this case, $w''=\#c^{n+1}a^{N}$ is not accepted by $\cD$ for sufficiently large $N \in \bbN$, which again contradicts $\lang(\cA,0)=\lang(\cD,0)$.\hfill \qed

\subsubsection{$\cB$ is $\zerodet$, but not $\foralldet$}
\label{app:zeroButNotForall}

We define formally $\cB=\tup{\{q_0,q',q''\},\{a,b,c,\#\},q_0,\delta_{\cB},\{q',q''\}}$, for:

$\delta_{\cB}=\left\{(q_0,\#.-1,q'),(q_0,\#.-1,q''),(q',a,1,q'),(q',b,0,q'),(q',c,-1,q')\right\} \cup \\
\left\{(q'',a,0,q''),(q'',b,1,q''),(q'',c,-1,q'')\right\}$.

Since $\lang(\cB,0)=\emptyset$, $\cB$ is $\zerodet$ trivially. However, since with initial counter 0, both $(q_0,\#.-1,q')$ and $(q_0,\#.-1,q'')$ cannot be traversed, we have that $\lang(\cB,1)=\lang(\cA,0)$. therefore, as can be shown by an identical analysis to the one presented in \cref{app:existsButNotZero}, there is no deterministic OCN $\cD$ that satisfies $\lang(\cB,1)=\lang(\cD,0)$, and $\cB$ is not $\foralldet$. \hfill \qed

\subsubsection{$\cC$ is $\foralldet$, but not $\unidet$}
\label{app:ForallButNotUni}

We define formally $\cC=\tup{\{q_0,q_1,q_2\},\{a,b,\#\},q_0,\delta_{\cC},\{q_1,q_2\}}$, for:

$\delta_{\cC}=\left\{(q_0,\#.0,q_1),(q_0,\#.-1,q_2),(q_1,a.1,q_1),(q_1,b,-1,q_1)\right\} \cup \\
\left\{(q_2,a,0,q_2),(q_2,b,0,q_2)\right\}$.

For initial counter $0$, the transition $(q_0,\#.-1,q_2)$ cannot be traversed, therefore $\cC$ is $\zerodet$, since $\cD=\tup{\{q_0,q_1\},\{a,b,\#\},q_0,\left\{(q_0,\#.0,q_1),(q_1,a.1,q_1),(q_1,b,-1,q_1)\right\},\{q_1\}}$ satisfies $\lang(\cD,0)=\lang(\cC,0)$. In addition, $\lang(\cC,k)=\#\{a,b\}^*$ for all $k \geq 1$. Hence $\cC$ is $\foralldet$.

Now assume by way of contradiction that $\cC$ is $\unidet$, and let $\cD$ be a deterministic OCN with $n \in \bbN$ states that satisfies $\lang(\cD,k)=\lang(\cC,k)$ for all $k \in \bbN$, and let $w=\#a^{n+1}b^{n+1} \in \lang(\cD,k)$ for all $k \in \bbN$. $\cD$ travels through a cycle $\beta$ when reading $b^{n+1}$. If the cumulative weight of $\beta$ is non-negative, then $w'=\#a^{n+1}b^{N} \in \lang(\cD,0)$ for arbitrarily large $N \in \bbN$, which contradicts $\lang(\cD,0)=\lang(\cC,0)$. If, however, the cumulative weight of $\beta$ is negative, then $w'=\#a^{n+1}b^{N} \notin \lang(\cD,1)$ for large enough $N \in \bbN$, which in turn contradicts $\lang(\cD,1)=\lang(\cC,1)$. \hfill \qed

%\section{Proofs of Section \ref{sec:general}}
%\label{app:general}

\subsection{Proof of \cref{lem:finFromZeroUndecidable}}
\label{app:finFromZeroUndecidable}
%\begin{proof}
We prove undecidability of $\finreachzero$ using a straightforward reduction from $\finreach$. Given an LCM $\cM=(\loc,\cC,\Delta)$ and a configuration $\sigma_{0}=\tup{q,(a_1,a_2 \ldots a_n)}$, we define an LCM $\cM'$ with a new initial state $q_0$ that leads to $q$ with a single path that increments $z_1$ $a_1$ times, $z_2$ $a_2$ times, etc. 

Formally, If $a_i=0$ for all $0 \leq i \leq n$, we define $\cM'=\cM$ and the reduction is trivial. Otherwise, we define $\cM'=(\loc',\cC,\Delta')$ such that:
\begin{itemize}
    \item $\loc'=\loc \cup \{q_0\} \cup \left\{\{q_i\}_{i=1}^{\Sigma_{j=1}^n a_j -1}\right\}$. Note that if $\Sigma_{j=1}^n a_j=1$, the only new state added is $q_0$.  
    \item $\Delta'=\Delta \cup \left\{(q_{\Sigma_{j=1}^n a_j -1}, (z_y, \inc), q)\right\} \cup \left\{(q_i, (z_x, \inc), q_{i+1})\right\}$ such that $y$ is the largest integer $0 \leq y \leq n$ for which $a_y \neq 0$, and the parameter $x$ varies such that throughout the $\Sigma_{j=1}^n a_j$ transitions, each counter $z_i$ is incremented exactly $a_i$ times.
\end{itemize}

The reachability set of $\sigma_{0}=\tup{q,(a_1,a_2 \ldots a_n)}$ under $\cM$ is finite iff the reachability sets of all configurations $\sigma_{0}'=\tup{q,(a_1',a_2' \ldots a_n')}$ such that $a_i' \leq a_i$ for all $i$ are finite, due to monotonicity of LCMs. This, in turn, is satisfied iff the reachability set of $\tup{q_0,(0,0 \ldots 0)}$ under $\cM'$ is finite. 
%\end{proof}
\hfill\qed

\subsection{Proof of \cref{lem:lcmToOcnConstruction}}
\label{app:lcmToOcnConstruction}

We start by describing several gadgets used in the construction.

\subsubsection{Gadgets} \label{sec:lcmToOcnSubs}
Let $\cM=\tup{\loc,\counters,\Delta}$ be an LCM, let $z_i \in \counters$, and let $(\ell_1,\op,\ell_2) \in \Delta$. Our goal is to construct an OCN $\cA$ that reads two consecutive configuration encodings - an encoding that corresponds to a visit in $\ell_2$ and then an encoding that corresponds to a visit in $\ell_1$, such that $w \in \lang(\cA,0)$ iff $w$ admits a violation for counter $z_i$. 

The structure of $\cA$ depends on the value of $\op$, which can any of the following:
\begin{enumerate}
    \item $z_i\inc$, i.e., increment $z_i$,
    \item $z_i\dec$, i.e., decrement $z_i$,
    \item $z_j\inc$ or $z_j\dec$ for $j\neq i$, which does not affect $z_i$,
    \item $z_i\jz$, i.e., test $z_i$ for 0.
\end{enumerate}
In addition, we have a special gadget to capture violations in the initial configuration, namely if the counter values is not 0 (recall that the initial configuration is read last, since the encoding is reversed).

Thus, $\cA$ can be any of the gadgets presented in figure \ref{fig:violationCheckers} (depending on $\op$). 

\begin{figure}[ht]
\captionsetup[subfigure]{justification=centering}
\begin{subfigure}[b]{.5\linewidth}
\centering \includegraphics[page=1]{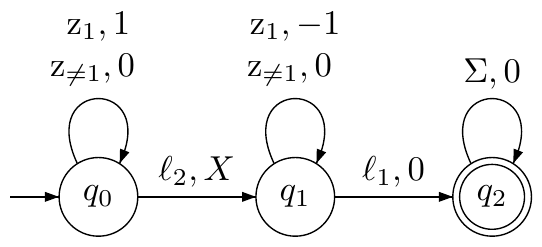}
\caption{Gadgets for scenarios 1,2, and 3, by setting $X$ to be $-2,0$, and $-1$, respectively.}%\label{fig:gadgetChecker}
\end{subfigure}%
\begin{subfigure}[b]{.5\linewidth}
\centering \includegraphics[page=2]{graphics/gadgets-pics.pdf}
\caption{Gadget for scenario 4.}%\label{fig:existsDetNotZeroDet}
\end{subfigure}\\
\begin{subfigure}[b]{.4\linewidth}
\centering \includegraphics[page=3]{graphics/gadgets-pics.pdf}
\caption{Gadget for initial configuration (last one in the reverse encoding).}%\label{fig:zeroDetnotForallDet}
\end{subfigure}%
\caption{The violation-check gadgets for $z_1$. By $z_{\neq 1}$ we mean $z_j$ for all $j\neq 1$, and by $z_j$ we mean every counter.}
\label{fig:violationCheckers}
\end{figure}

Formally, we define $\cA=\tup{\Sigma,\{q_0,q_1,q_2\},q_0,\delta,\{q_2\}}$ such that:
\begin{itemize}
    \item $\Sigma=\loc \cup \left\{\left\{a_i\right\}_{z_i \in \counters}\right\}$.
    \item $\delta=\left\{(q_0,a_j,0,q_0)\right\}_{j \neq i} \cup  \left\{(q_0,a_i,1,q_0)\right\} \cup 
    \left\{(q_0,\ell_2,\nu,q_1)\right\} \cup 
    \left\{(q_1,a_j,0,q_1)\right\}_{j \neq i} \cup 
    \left\{(q_1,a_i,-1,q_1)\right\} \cup 
    \left\{(q_1,\ell_1,0,q_2)\right\} \cup 
    \left\{(q_2,\sigma,0,q_2)\right\}_{\sigma \in \Sigma}$.
\end{itemize}

%This definition matches scenarios 1,2,3 for $\nu=-2,0,-1$ respectively. In order to handle For scenario 4, the gadget is almost identical, except we define $\nu=-1$, remove the transition $(q_1,a_i,-1,q_1)$ from $\delta$, and add the transition $(q_1,a_i,0,q_1)$ instead.

%, we formally define the violation capturer aimed to read the first configuration of a run of $\cM$ and accept if any of the counter values is not 0. This structure is also illustrated in figure \ref{fig:violationCheckers}.
For the initial configuration checker, we define $\cA=\tup{\Sigma,\{q_0,q_1\},q_0,\delta,\{q_1\}}$ such that:
\begin{itemize}
    \item $\Sigma=\loc \cup \left\{a_i\right\}_{z_i \in \counters}$.
    \item $\delta=\left\{(q_0,a_j,1,q_0)\right\}_{z_j \in \counters} \cup        \left\{(q_0,\ell_0,-1,q_1)\right\}$.
\end{itemize}

Our last gadget captures ill-formed words, regardless of counter values.

Let $\cM=\tup{\loc,\counters,\Delta}$ be an LCM. we say that a word $w$ is \emph{well formed} if the following conditions are satisfied:
\begin{enumerate}
    \item $w$ is of the form $w=a_n^*\cdots a_1^*\ell_{iN} \cdots a_n^*\cdots a_1^* \ell_{i0}$ for $\left\{\ell_{ij} \in \loc\right\}_{0 \leq j \leq N}$.
    \item $\ell_{i0}=\ell_0$.
    \item for every $0 \leq j \leq N-1$, there is at least one transition in $\cM$ that leads from $\ell_{ij}$ to $\ell_{i,j+1}$.
\end{enumerate}

It is easy to see that well formed words are a regular language, and in particular its complement is the desired OCN.

\subsubsection{The Main Construction}

Let $\cM=\tup{\loc,\counters,\Delta}$. We wish to construct an OCN $\cA$ such that $\lang(\cA,0)$ is the set of all words that do not represent legal runs of $\cM$.

Intuitively, we construct $\cA$ through the following process:
\begin{enumerate}
    \item Construct a flow violation checker (with regards to $\cM$), which will be part of $\cA$ as a separate component.
    \item for every location $\ell \in \loc$, add a corresponding state $\ell'$ in $\cA$. all such $\ell'$'s are initial states in $\cA$, and they all have self loops with weight 0 when reading all counter accumulators $\left\{a_i\right\}_{z_i \in \counters}$. Intuitively, when $\cA$ visits a state $\ell'$, it means that $\cA$ is currently in the process of reading a configuration in which $\cM$ is in location $\ell$.
    \item for every transition $(\ell_1, \op, \ell_2) \in \Delta$, add the transition $(\ell_2', \ell_2, 0, \ell_1')$ to $\cA$. Intuitively, traveling this transition means that $\cA$ has finished reading a configuration of location $\ell_2$, and is now starting to read a configuration of location $\ell_1$. 
    \item connect an initial configuration violation checker to $\ell_0'$.
    \item for every transition $(\ell_1, \op, \ell_2) \in \Delta$, add from $\ell_2'$  transitions to all relevant violation checkers for all counters $\left\{ z_i\right\}_{1 \leq i \leq n}$.   
\end{enumerate}

Now let us define the construction formally. Let $V(\ell_i \rightarrow \ell_j, z_m)$ be the violation checker that matches the transition $(\ell_i, \op, \ell_j)$ for counter $z_m$, as detailed in \cref{sec:lcmToOcnSubs}. Let $Q(\ell_i \rightarrow \ell_j, z_m)$ be its states, let $F(\ell_i \rightarrow \ell_j, z_m)$ be its accepting states, let $\delta(\ell_i \rightarrow \ell_j, z_m)$ be its transitions, and $\lambda(\ell_i \rightarrow \ell_j, z_m) \subseteq \delta(\ell_i \rightarrow \ell_j, z_m)$ be the
 transitions from its initial state. In that spirit we also define, with regards to the flow control violation checker, and the initial configuration violation checker: $Q(\text{initial})$, $\delta(\text{initial})$, $\lambda(\text{initial})$, $Q(\text{flow})$, $\delta(\text{flow})$, $\lambda(\text{flow})$. Lastly, for convenience' sake alone we define $\cA$ as having multiple initial states. this has been done for readability, and can easily be formally circumvented by defining a single initial state $\alpha_0$, along with an outgoing transition $(\alpha_0,\sigma,z,q)$ for each $(s_0,\sigma,z,q) \in \delta$.

We now define $\cA=\tup{\Sigma,Q,S_0,\delta,F}$ such that:
\begin{itemize}
    \item $\Sigma = \loc \cup \left\{a_i\right\}_{z_i \in \counters}$
    \item $Q=\left\{\ell_i'\right\}_{\ell_i \in \loc} \cup Q(\text{initial}) \cup Q(\text{flow}) \cup \left\{Q(\ell_i \rightarrow \ell_j, z_m)\right\}$ for all $\ell_i,\ell_j \in \loc$ such that there is a transition from $\ell_i$ to $\ell_j$ in $\Delta$, and for all $1 \leq i \leq m$.
    \item $S_0=\left\{\ell_i'\right\}_{\ell_i \in \loc} \cup \{s_{0,\text{flow}}\}$ such that $s_{0,\text{flow}}$ is the initial state of the flow violation checker.
    $\delta_1=\left\{(\ell_i',a_j,0,\ell_i')\right\}$ for all $\ell_i \in \loc$ and for all $1 \leq j \leq n$.
    \item $\delta_2=\left\{(\ell_i',\ell_i,0,\ell_j')\right\}$ for all $\ell_i,\ell_j \in \loc$ such that there is a transition from $\ell_j$ to $\ell_i$ in $\Delta$.
    \item $\delta_3=\left\{(\ell_i',\sigma,\nu,q')\right\}$ for all $\ell_i,\ell_j \in \loc$ such there is a transition from $\ell_j$ to $\ell_i$ in $\Delta$, and $(q,\sigma,\nu,q')\in \lambda(\ell_j \rightarrow \ell_i, z_m)$ for some $1 \leq i \leq m$, or otherwise $(q,\sigma,\nu,q')\in \lambda(\text{initial})$.
    \item $\delta_V= \bigcup_{\text{all violations}}\delta(\text{violation})$. 
    \item $\delta=\delta_1 \cup \delta_2 \cup \delta_3 \cup \delta_V$
    \item $F=\bigcup_{\text{all violations}}F(\text{violation})$.
\end{itemize}

We turn to prove the correctness of the construction. 
Consider a word $w$ that represents a legal run of $\cM$. Then, first of all, $w$ is well formed, and therefore not accepted by the flow violation checker. second, there is no transition from one configuration to the next that involves a violation, and therefore $w$ cannot be accepted through any of the violation checkers in $\cA$. Since all accepting states of $\cA$ are inside violation checkers, $w \notin \lang(\cA,0)$.

Conversely, assume a word $w$ does not represent a legal run of $\cM$. If $w$ is not well formed, then it is accepted through the flow violation checker. Otherwise - a transition from a state $\ell_i \in \loc$ to a state $\ell_j \in \loc$ represents a violation for counter $z_m$ such that $1 \leq m \leq n$. $\cA$ then accepts $w$ by branching from $\ell_j'$ to $V(\ell_i \rightarrow \ell_j, z_m)$ at the right moment. It is also possible that the violation occurs in the first configuration (last one to be read), and in this case $w$ will be accepted through the initial configuration violation checker.

\subsection{Construction of $\cD$ in \cref{lem:finReachThendet}}
\label{app:finReachThendet}
\shtodo{I changed $q$ to $\ell$ and $Q$ to $\loc$ and $a_j$ to $z_j$ for the letters. Please go over this to make sure there are no mistakes.}
We define $\cD=\tup{\Sigma,Q',s_0',\delta',F'}$, such that:
\begin{itemize}
    \item $Q'=\left\{\tup{\ell,a_{1} \ldots a_{k}, b_{1} \ldots b_{k}} | \ell \in \loc, 0 \leq a_{i},b_{i} \leq m \text{ for all } 1 \leq i \leq k \right\} \cup \\
    \left\{\tup{\bot,\bot \ldots \bot , b_{1} \ldots b_{k}} | 0 \leq b_{i} \leq m \text{ for all } 1 \leq i \leq k \right\}$.
    \item $s_0=\tup{\bot,\bot \ldots \bot , 0 \ldots 0}$.
    \item $\delta'(\tup{\ell,a_{1} \ldots a_{k}, b_{1} \ldots b_{k}},\ell')=\tup{\ell',b_{1} \ldots b_{k} , 0 \ldots 0}$ if the configuration $\tup{\ell,a_{1} \ldots a_{k}}$ can be obtained from the configuration $\tup{\ell',b_{1} \ldots b_{k}}$ through a single transition in $\cM$.
    \item $\delta'(\tup{\bot,\bot \ldots \bot , b_{1} \ldots b_{k}},\ell)=\tup{\ell,b_{1} \ldots b_{k} , 0 \ldots 0}$ for all $\ell \in \loc, 0 \leq b_{1} \ldots b_{k} \leq m$.
    \item $\delta'(\tup{\ell,a_{1} \ldots a_{k} , 0 \ldots 0,b_{j} \ldots b_{k}},z_{j})=\tup{\ell,a_{1} \ldots a_{k} , 0 \ldots 0,b_{j}+1 \ldots b_{k}}$ for all $0 \leq j \leq k$, $b_{j} < m$. 
    \item $\delta'(\tup{\ell,a_{1} \ldots a_{k} , 0 \ldots 0,b_{j} \ldots b_{k}},z_{j-x})=\tup{\ell,a_{1} \ldots a_{k} , 0 \ldots 1,0 \ldots b_{j} \ldots b_{k}}$ for all $1 \leq j \leq k$, $1 \leq x \leq j$. 
    \item $F=\{\tup{\ell_0,0 \ldots 0 , 0 \ldots 0}\}$.
\end{itemize}

%\section{Proofs of Section \ref{sec:singleton}}
%\label{app:singleton}

\subsection{Details for the Proof of \cref{lem:semilinearAndIncreadingThenPeriodic}}
\label{app:semilinearAndIncreadingThenPeriodic}
% We start by demonstrating our method, followed by the general construction.
% Consider, for example, $S=\lin((1,0),(4,8)) \cup \lin((2,1),(6,12))$. In this case $\gamma=12$. 
% We split $\lin((1,0),(4,8))$ to $\lin((1,0),(12,24)) \cup \lin((5,8),(12,24)) \cup \lin((9,16),(12,24))$, the intuition being that instead of a $(4,8)$ period, we have a $(12,24)$ period, and we add different basis vectors to fill the gaps, so the new basis vectors are $(5,8)$ and $(9,16)$, where the next basis vector $(13,24)$ is already captured by $(1,0)+(12,24)$.
% Similarly, we split $\lin((2,1),(6,12))$ to $\lin((2,1),(12,24)) \cup \lin((8,13),(12,24))$. 
% Overall we get $S=\bigcup_{v \in V}\lin(v,(12,24))$ for $V=\left\{(1,0),(5,8),(9,16),(2,1),(8,13)\right\}$.  

Generally, let $\gamma=\textrm{lcm}(\{p_i\}_{i=1}^M)$. We split each linear component $\lin((a_i,b_i), \{(p_i,r_i)\})$ 
to $\frac{\gamma}{p_i}$ parts, by defining the \emph{$\gamma$-split} of $\lin((a_i,b_i), (p_{i},r_{i}))$ (defined only for $p_i|\gamma$) to be 
%l-split$(lin(c_i,(l_{i},n_{i})) = 
$\bigcup_{i=0}^{\frac{\gamma}{p_i}-1}\lin((a_i,b_i)+
i \cdot (p_{i},r_{i}), (\gamma,r_{i}) \cdot \frac{\gamma}{r_i})$.
each such split is semilinear by definition, and it is straightforward to show that $S=\bigcup_{i=1}^k \text{l-split}(\lin((a_i,b_i), (p_{i},r_{i}))$.

\end{document}